\documentclass[letterpaper]{article}
\usepackage{amsmath,amsfonts,amsthm,amssymb}

\theoremstyle{plain}

\numberwithin{equation}{section}
\newtheorem{thm}{Theorem}[section]
\newtheorem{lem}[thm]{Lemma}
\newtheorem{cor}[thm]{Corollary}

\newcounter{cond}

\newcommand{\complex}{{\mathbb C}}
\newcommand{\positive}{{\mathbb N}}
\newcommand{\real}{{\mathbb R}}
\newcommand{\ascript}{{\mathcal A}}
\newcommand{\bscript}{{\mathcal B}}
\newcommand{\rmre}{\mathrm{Re\,}}
\newcommand{\rmtr}{\mathrm{tr}}
\newcommand{\rmspan}{\mathrm{span}}
\newcommand{\fhat}{\widehat{f}}
\newcommand{\ghat}{\widehat{g}}
\newcommand{\hhat}{\widehat{h}}
\newcommand{\cbar}{\bar{c}}
\newcommand{\fbar}{\bar{f}}
\newcommand{\ubar}{\bar{u}}
\newcommand{\chihat}{\widehat{\chi}}

\newcommand{\ab}[1]{\left|#1\right|}
\newcommand{\doubleab}[1]{\left\|#1\right\|}
\newcommand{\brac}[1]{\left\{#1\right\}}
\newcommand{\paren}[1]{\left(#1\right)}
\newcommand{\sqbrac}[1]{\left[#1\right]}
\newcommand{\elbows}[1]{{\left\langle#1\right\rangle}}
\newcommand{\ket}[1]{{\left|#1\right>}}
\newcommand{\bra}[1]{{\left<#1\right|}}

\begin{document}

\title{QUANTUM MEASURES\\and INTEGRALS
}
\author{S. Gudder\\ Department of Mathematics\\
University of Denver\\ Denver, Colorado 80208, U.S.A.\\
sgudder@.du.edu\\
}
\date{}
\maketitle

\begin{abstract}
We show that quantum measures and integrals appear naturally in any $L_2$-Hilbert space $H$. We begin by defining a decoherence operator $D(A,B)$ and it's associated $q$-measure operator $\mu (A)=D(A,A)$ on $H$. We show that these operators have certain positivity, additivity and continuity properties. If $\rho$ is a state on $H$, then $D_\rho (A,B)=\rmtr\sqbrac{\rho D(A,B)}$ and $\mu _\rho (A)=D_\rho (A,A)$ have the usual properties of a decoherence functional and $q$-measure, respectively. The quantization of a random variable $f$ is defined to be a certain self-adjoint operator $\fhat$ on $H$. Continuity and additivity properties of the map $f\mapsto\fhat$ are discussed. It is shown that if $f$ is nonnegative, then $\fhat$ is a positive operator. A quantum integral is defined by
$\int fd\mu _\rho =\rmtr (\rho\fhat\,)$. A tail-sum formula is proved for the quantum integral. The paper closes with an example that illustrates some of the theory.
\end{abstract}

\medskip
\noindent
{\bf Keywords:} quantum measures, quantum integrals, decoherence functionals.

\section{Introduction}  
Quantum measure theory was introduced by R.~Sorkin in his studies of the histories approach to quantum gravity and cosmology \cite{sor941, sor942}. Since 1994 a considerable amount of literature has been devoted to this subject
\cite{dgt08, gt09, gud101, mocs05, sal02, sor071, sw10} and more recently a quantum integral has been introduced \cite{gud091, gud092}. At first sight this theory appears to be quite specialized and its applicability has been restricted to the investigation of quantum histories and the related coevent interpretation of quantum mechanics
\cite{gtw09, gud092, gud102, sor072}. However, this article intends to demonstrate that quantum measure theory may have wider application and that its mathematical structure is already present in the standard quantum formalism. One of our aims is to show that quantum measures are abundant in any $L_2$-Hilbert space $H$. In particular, for any state (density operator) $\rho$ on $H$ there is a naturally associated quantum measure $\mu _\rho$. For an event
$A$ we interpret $\mu _\rho (A)$ as the quantum propensity that $A$  occurs. Moreover corresponding to $\rho$ there is a natural quantum integral $\int fd\mu _\rho$ that can be interpreted as the quantum expectation of the random variable $f$.

The article begins by defining a decoherence operator $D(A,B)$ for events $A,B$ and the associated $q$-measure operator $\mu (A)=D(A,A)$ on $H$. It is shown that these operator-valued functions have certain positivity, additivity and continuity properties. Of particular importance is the fact that although $\mu _\rho (A)$ is not additive, it does satisfy a more general grade-2 additivity condition. If $\rho$ is a state on $H$, then
$D_\rho (A,B)=\rmtr\sqbrac{\rho D(A,B)}$ and $\mu _\rho (A)=D_\rho (A,A)$ have the usual properties of a decoherence functional and $q$-measure, respectively. The quantization of a random variable $f$ is defined to be a certain self-adjoint operator $\fhat$ on $H$. Continuity and additivity properties of the map $f\mapsto\fhat$ are discussed. It is shown that if $f$ is nonnegative, then $\fhat$ is a positive operator. A quantum integral is defined by
$\int fd\mu _\rho =\rmtr(\rho\fhat )$. A tail-sum formula is proved for the quantum integral. It follows that
$\int fd\mu _\rho$ coincides with the quantum integral considered in previous works. The paper closes with an example that illustrates some of the theory. The example shows that the usual decoherence functionals and
$q$-measures considered before reduce to the form given in Section~2.

\section{Quantum Measures} 
A \textit{probability space} is a triple $(\Omega ,\ascript ,\nu )$ where $\Omega$ is a \textit{sample space} whose elements are \textit{sample points} or \textit{outcomes}, $\ascript$ is a $\sigma$-algebra of subsets of $\Omega$ called \textit{events} and $\nu$ is a measure on $\ascript$ satisfying $\nu (\Omega )=1$. For $A\in\ascript$,
$\nu (A)$ is interpreted as the probability that event $A$ occurs. Let $H$ be the Hilbert space
\begin{equation*}
H=L_2(\Omega ,\ascript ,\nu )=\brac{f\colon\Omega\to\complex ,\int\ab{f}^2d\nu <\infty}
\end{equation*}
with inner product $\elbows{f,g}=\int\fbar gd\nu$. We call real-valued functions $f\in H$ \textit{random variables}. If
$f$ is a random variable, then by Schwarz's inequality
\begin{equation}         
\label{eq21}
\ab{\int fd\nu}\le\int\ab{f}d\nu\le\doubleab{f}
\end{equation}
so the \textit{expectation} $E(f)=\int fd\nu$ exists and is finite. Of course \eqref{eq21} holds for any $f\in H$.

The characteristic function $\chi _A$ of $A\in\ascript$ is a random variable with $\doubleab{\chi _A}=\nu (A)^{1/2}$ and we write $\chi _\Omega =1$. For $A,B\in\ascript$ we define the \textit{decoherence operator} $D(A,B)$ as the operator on $H$ defined by $D(A,B)=\ket{\chi _A}\bra{\chi _B}$. Thus, for $f\in H$ we have
\begin{equation*}
D(A,B)f=\elbows{\chi _B,f}\chi _A=\int _Bfd\nu\chi _A
\end{equation*}
Of course, if $\nu (A)\nu (B)=0$, then $D(A,B)=0$.

\begin{lem}       
\label{lem21}
If $\nu (A)\nu (B)\ne 0$, then $D(A,B)$ is a rank~1 operator with $\doubleab{D(A,B)}=\nu (A)^{1/2}\nu (B)^{1/2}$.
\end{lem}
\begin{proof}
It is clear that $D(A,B)$ is a rank~1 operator with range $\rmspan (\chi _A)$. For $f\in H$ we have
\begin{align*}
\doubleab{D(A,B)f}&=\doubleab{\elbows{\chi _B,f}\chi _A}=\ab{\elbows{\chi _B,f}}\,\doubleab{\chi _A}\\
  &\le\doubleab{\chi _A}\,\doubleab{\chi _B}\,\doubleab{f}=\nu (A)^{1/2}\nu (B)^{1/2}\doubleab{f}
\end{align*}
Hence, $\doubleab{D(A,B)}\le\nu (A)^{1/2}\nu (B)^{1/2}$. Letting $g$ be the unit vector $\chi _B/\doubleab{\chi _B}$ we have
\begin{equation*}
\doubleab{D(A,B)g}=\ab{\elbows{\chi _B,g}}\,\doubleab{\chi _A}=\doubleab{\chi _B}\,\doubleab{\chi _A}
  =\nu (A)^{1/2}\nu (B)^{1/2}
\end{equation*}
The result now follows.
\end{proof}

For $A\in\ascript$ we define the $q$-\textit{measure operator} $\mu (A)$ on $H$ by
\begin{equation*}
\mu (A)=D(A,A)=\ket{\chi _A}\bra{\chi _A}
\end{equation*}
Hence,
\begin{equation*}
\mu (A)f=\elbows{\chi _A,f}\chi _A=\int _Afd\nu\chi _A
\end{equation*}
If $\nu (A)\ne 0$, then by Lemma~\ref{lem21}, $\mu (A)$ is a positive (and hence self-adjoint) rank~1 operator with
$\doubleab{\mu (A)}=\nu (A)$.

We now show that $A\mapsto\chi _A$ is a \textit{vector-valued measure} on $\ascript$. Indeed, if
$A\cap B=\emptyset$, then $\chi _{A\cup B}=\chi _A+\chi _B$ so $A\mapsto\chi _A$ is \textit{additive}. Moreover, if
$A_1\subseteq A_2\subseteq\cdots$ is an increasing sequence of events, then letting $A=\cup A_i$ we have
\begin{equation*}
\doubleab{\chi _A-\chi _{A_n}}^2=\doubleab{\chi _{A\smallsetminus A_n}}^2=\nu (A\smallsetminus A_n)
  =\nu (A)-\nu (A_n)\to 0
\end{equation*}
Hence, $\lim\chi _{A_n}=\chi _{\cup A_i}$. The \textit{countable additivity condition}
\begin{equation*}
\chi _{\cup B_i}=\sum\chi _{B_i}
\end{equation*}
follows for mutually disjoint $B_i\in\ascript$ where the convergence of the sum is in the vector norm topology. Since
$\chi _A$ is orthogonal to $\chi _B$ whenever $A\cap B=\emptyset$, we call $A\mapsto\chi _A$ an
\textit{orthogonally scattered} vector-valued measure. A similar computation shows that if
$A_1\supseteq A_2\supseteq\cdots$ is a decreasing sequence on $\ascript$, then 
\begin{equation*}
\lim\chi _{A_n}=\chi _{\cap A_i}
\end{equation*}
This also follows from the fact that the complements $A'_i$ form an increasing sequence so by additivity
\begin{align*}
\lim\chi _{A_n}&=1-\lim\chi _{A'_n}=1-\chi _{\cup A'_i}=1-\sqbrac{\chi _{(\cap A_i)'}}\\
  &=\chi _{\cap A_i}
\end{align*}

The map $D$ from $\ascript\times\ascript$ into the set of bounded operators $\bscript (H)$ on $H$ has some obvious properties:
\begin{list} {(\arabic{cond})}{\usecounter{cond}
\setlength{\rightmargin}{\leftmargin}}
\item If $A\cap B=\emptyset$, then $D(A\cup B,C)=D(A,C)+D(B,C)$ for all $C\in\ascript$ (additivity)
\item $D(A,B)^*=D(B,A)$ (conjugate symmetry)
\item $D(A,B)^2=\nu (A\cap B)D(A,B)$
\item $D(A,B)D(A,B)^*=\nu (B)\mu (A)$, $D(A,B)^*D(A,B)=\nu (A)\mu (B)$
\end{list}
Less obvious properties are given in the following theorem.

\begin{thm}       
\label{thm22}
{\rm (a)}\enspace $D\colon\ascript\times\ascript\to\bscript (H)$ is positive semidefinite in the sense that if
$A_i\in\ascript$, $c_i\in\complex$, $i=1,\ldots ,n$, then
\begin{equation*}
\sum _{i,j=1}^nD(A_i,A_j)c_i\cbar _j
\end{equation*}
is a positive operator.
{\rm (b)}\enspace If $A_1\subseteq A_2\subseteq\cdots$ is an increasing sequence in $\ascript$, then the continuity condition
\begin{equation*}
\lim D(A_i,B)=D(\cup A_i,B)
\end{equation*}
holds for every $B\in\ascript$ where the limit is in the operator norm topology.
\end{thm}
\begin{proof}
(a)\enspace For $A_i\in\ascript$, $c_i\in\complex$, $i=1,\ldots ,n$, we have
\begin{equation}         
\label{eq22}
\sum _{i,j=1}^nD(A_i,A_j)c_i\cbar _j=\sum _{i,j=1}^n\ket{\chi _{A_i}}\bra{\chi _{A_j}}c_i\cbar _j
  =\ket{\sum _{i=1}^nc_i\chi _{A_i}}\bra{\sum _{j=1}^nc_j\chi _{A_j}}
\end{equation}
Since the right side of \eqref{eq22} is a positive operator, the result follows.\newline
(b)\enspace For the increasing sequence $A_i$, let $A=\cup A_i$ and let $f\in H$. Then
\begin{align*}
\doubleab{\sqbrac{D(A,B)-D(A_i,B)}f}&=\doubleab{\int _Bfd\nu (\chi _A-\chi _{A_i}}
  =\ab{\int _Bfd\nu}\,\doubleab{\chi _A-\chi _{A_i}}\\
  &=\ab{\int _Bfd\nu}\sqbrac{\nu (A)-\nu (A_i)}^{1/2}\\
  &\le\int _B\ab{f}d\nu\sqbrac{\nu (A)-\nu (A_i)}^{1/2}\\
  &\le\sqbrac{\nu (A)-\nu (A_i)}^{1/2}\doubleab{f}
\end{align*}
Hence
\begin{equation*}
\lim\doubleab{D(A,B)-D(A_i,B)}\le\lim\sqbrac{\nu (A)-\nu (A_i)}^{1/2}=0
\qedhere
\end{equation*}
\end{proof}

If $A_i$ are mutually disjoint events, the countable additivity condition
\begin{equation*}
D\paren{\bigcup _{i=1}^\infty A_i,B}=\sum _{i=1}^\infty D(A_i,B)
\end{equation*}
follows from Theorem~\ref{thm22}(b). We conclude that $A\mapsto D(A,B)$ is an operator-valued measure. By conjugate symmetry, $B\mapsto D(A,B)$ is also an operator-valued measure. As before, it follows that if $A_1\supseteq A_2\supseteq\cdots$ is a decreasing sequence in $\ascript$, then
\begin{equation*}
\lim D(A_i,B)=D(\cap A_i,B)
\end{equation*}
for all $B\in\ascript$.

The map $\mu\colon\ascript\to\bscript (H)$ need not be additive. For example, if $A,B\in\ascript$ are disjoint, then
\begin{align*}
\mu (A\cup B)&=\ket{\chi _{A\cup B}}\bra{\chi _{A\cup B}}=\ket{\chi _A+\chi _B}\bra{\chi _A+\chi _B}\\
  &=\ket{\chi _A}\bra{\chi _A}+\ket{\chi _B}\bra{\chi _B}+\ket{\chi _A}\bra{\chi _B}+\ket{\chi _B}\bra{\chi _A}\\
   &=\mu (A)+\mu (B)+2\rmre D(A,B)
\end{align*}
Notice that additivity is spoiled by the presence of the self-adjoint operator $2\rmre D(A,B)$. For this reason, we view this operator as measuring the interference between the events $A$ and $B$. Because of this nonadditivity, we have that $\mu (A')\ne\mu (\Omega )-\mu (A)$ in general and $A\subseteq B$ need not imply $\mu (A)\ne\mu (B)$ in the usual order of self-adjoint operators. However, $A\mapsto\mu (A)$ does satisfy the condition given in (a) of the next theorem.

\begin{thm}       
\label{thm23}
{\rm (a)}\enspace $\mu$ satisfies grade-2 additivity:
\begin{equation*}
\mu (A\cup B\cup C)=\mu (A\cup B)+\mu (A\cup C)+\mu (B\cup C)-\mu (A)-\mu (B)-\mu (C)
\end{equation*}
whenever $A,B,C\in\ascript$ are mutually disjoint.
{\rm (b)}\enspace $\mu$ satisfies the continuity conditions
\begin{align*}
\lim\mu (A_i)&=\mu (\cup A_i)\\
\lim\mu (B_i)&=\mu (\cap A_i)
\end{align*}
in the operator norm topology for any increasing sequence $A_i$ in $\ascript$ or decreasing sequence
$B_i\in\ascript$.
\end{thm}
\begin{proof}
(a)\enspace For $A,B,C\in\ascript$ mutually disjoint, we have
\begin{align*}
\mu (A&\cup B)+\mu (A\cup C)+\mu (B\cup C)-\mu (A)-\mu (B)-\mu (C)\\
  &=\ket{\chi _A+\chi _B}\bra{\chi _A+\chi _B}+\ket{\chi _A+\chi _C}\bra{\chi _A+\chi _C}
  +\ket{\chi _B+\chi _C}\bra{\chi _B+\chi _C}\\
  &\quad -\ket{\chi _A}\bra{\chi _A}-\ket{\chi _B}\bra{\chi _B}-\ket{\chi _C}\bra{\chi _C}\\
  &=\ket{\chi _A}\bra{\chi _A}+\ket{\chi _B}\bra{\chi _B}+\ket{\chi _C}\bra{\chi _C}+\ket{\chi _A}\bra{\chi _B}
  +\ket{\chi _B}\bra{\chi _A}\\
  &\quad +\ket{\chi _A}\bra{\chi _C}+\ket{\chi _C}\bra{\chi _A}+\ket{\chi _B}\bra{\chi _C}+\ket{\chi _C}\bra{\chi _C}\\
  &=\ket{\chi _A+\chi _B+\chi _C}\bra{\chi _A+\chi _B+\chi _C}=\mu (A\cup B\cup C)
\end{align*}
(b)\enspace For an increasing sequence $A_i\in\ascript$, let $A=\cup A_i$. For $f\in L_2(\Omega ,\ascript ,\nu )$ we have
\begin{align*}
\doubleab{\sqbrac{\mu (A_i)-\mu (A)}f}&=\doubleab{\int _{A_i}fd\nu\chi _{A_i}-\int _Afd\nu\chi _A}\\
  &\le\doubleab{\int _{A_i}fd\nu\chi _{A_i}-\int _{A_i}fd\nu\chi _A}
  +\doubleab{\int _{A_i}\!fd\nu\chi _A-\int _A\!fd\nu\chi _A}\\
  &=\doubleab{\int _{A_i}fd\nu (\chi _{A_i}-\chi _A)}+\doubleab{\int f(\chi _{A_i}-\chi _A)d\nu\chi _A}\\
  &=\ab{\int _{A_i}fd\nu}\,\doubleab{\chi _A-\chi _{A_i}}+\ab{\int f(\chi _A-\chi _{A_i})d\nu}\,
  \doubleab{\chi _A}\\
  &\le\int _{A_i}\ab{f}d\nu\sqbrac{\nu (A)-\nu (A_i)}^{1/2}+\doubleab{f}\,\doubleab{\chi _A-\chi _{A_i}}\nu (A)^{1/2}\\
  &\le 2\nu (A)^{1/2}\sqbrac{\nu (A)-\nu (A_i)}^{1/2}\doubleab{f}
\end{align*}
Hence,
\begin{equation*}
\doubleab{\mu (A_i)-\mu (A)}\le 2\nu (A)^{1/2}\sqbrac{\nu (A)-\nu (A_i)}^{1/2}\to 0
\end{equation*}
A similar proof holds for a decreasing sequence $B_i\in\ascript$
\end{proof}

Additional properties of $\mu$ are given in the next lemma.

\begin{lem}       
\label{lem24}
{\rm (a)}\enspace If $A\cap B=\emptyset$, then $\mu (A)\mu (B)=0$.
{\rm (b)}\enspace $\mu (A)=0$ if and only if $\nu (A)=0$.
{\rm (c)}\enspace If $A\cap B=\emptyset$ and $\mu (A)=0$, then $\mu (A\cup B)=\mu (B)$.
{\rm (d)}\enspace If $\mu (A\cup B)=0$, then $\mu (A)=\mu (B)=0$.
\end{lem}
\begin{proof}
That (a) holds is clear.
(b)\enspace If $\mu (A)=0$, then for any $f\in H$ we have
\begin{equation*}
\int _Afd\nu\chi _A=\mu (A)f=0
\end{equation*}
Letting $f=1$ gives $\nu (A)\chi _A=0\hbox{ a.e. } [\nu ]$. Hence, $\nu (A)=0$. The converse clearly holds.
(c)\enspace If $A\cap B=\emptyset$ and $\mu (A)=0$, then by (b) we have that $\nu (A)=0$. Hence, for $f\in H$ we have
\begin{equation*}
\ket{\chi _A}\bra{\chi _B}f=\int _Bfd\nu\chi _A=0
\end{equation*}
We conclude that
\begin{equation*}
\mu (A\cup B)=2\rmre\ket{\chi _A}\bra{\chi _B}+\ket{\chi _B}\bra{\chi _B}=\mu (B)
\end{equation*}
(d)\enspace If $\mu (A\cap B)=0$, then by (b) we have that $\nu (A\cap B)=0$. Since $\nu$ is additive,
$\nu (A)=\nu (B)=0$ so by (b), $\mu (A)=\mu (B)=0$.
\end{proof}
 If $\rho$ is a density operator (state) on $H$, we define the \textit{decoherence functional}
 $D _\rho\colon\ascript\times\ascript\to\complex$ by
 \begin{equation*}
D_\rho (A,B)=\rmtr\sqbrac{\rho D(A,B)}=\elbows{\rho\chi _A,\chi _B}
\end{equation*}
Decoherence functionals have been extensively studied in the literature \cite{dgt08, gt09, mocs05, sor942, sor071} where $D_\rho (A,B)$ is used to describe the interference between $A$ and $B$ for the state $\rho$. Concrete examples of $D_\rho (A,B)$ are given in Section~3. Notice that
\begin{equation}         
\label{eq23}
\ab{D_\rho (A,B)}\le\doubleab{\rho}\,\doubleab{\chi _A}\,\doubleab{\chi _B}
  \le\nu (A)^{1/2}\nu (B)^{1/2}
\end{equation}
The next result, which follows from Theorem~\ref{thm22}, shows that $D_\rho (A,B)$ has the usual properties of a decoherence functional.

\begin{cor}       
\label{cor25}
{\rm (a)}\enspace $A\mapsto D_\rho (A,B)$ is a complex measure on $\ascript$ for any $B\in\ascript$.
{\rm (b)}\enspace If $A_1,\ldots ,A_n\in\ascript$, then the $n\times n$ matrix $D_\rho (A_i,A_j)$ is positive semidefinite. 
\end{cor}

For a density operator $\rho$ on $H$, we define the $q$-\textit{measure} $\mu _\rho\colon\ascript\to\real ^+$ by
\begin{equation*}
\mu _\rho (A)=\rmtr\sqbrac{\rho\mu (A)}=\elbows{\rho\chi _A,\chi _A}
\end{equation*}
We interpret $\mu _\rho (A)$ as the quantum propensity that the event $A$ occurs \cite{gud101, sor941, sor942}. It follows from \eqref{eq23} that $\mu _\rho (A)\le\nu (A)$. Theorem~\ref{thm23} holds with $\mu$ replaced by
$\mu _\rho$. This shows that $\mu _\rho$ has the usual properties of a $q$-measure.

\section{Quantum Integrals} 
Let $f\in H$ be a nonnegative random variable. The \textit{quantization} of $f$ is the operator $\fhat$ on $H$ defined by
\begin{equation}         
\label{eq31}
(\fhat g)(y)=\int\min\sqbrac{f(x),f(y)}g(x)d\nu (x)
\end{equation}
We can write \eqref{eq31} as
\begin{equation*}
(\fhat g)(y)=\int _{\brac{x\colon f(x)\le f(y)}}fgd\nu +f(y)\int _{\brac{x\colon f(x)>f(y)}}gd\nu
\end{equation*}
Since
\begin{align*}
\ab{(\fhat g)(y)}&\le\int\min\sqbrac{f(x),f(y)}\ab{g(x)}d\nu (x)\le\int f\ab{g}d\nu\\
  &\le\doubleab{f}\,\doubleab{g}
\end{align*}
we have $\doubleab{\fhat g}\le\doubleab{f}\,\doubleab{g}$. We conclude that $\fhat$ is bounded with
$\doubleab{\fhat}\le\doubleab{f}$. Since $\fhat$ is bounded and symmetric, it follows that $\fhat$ is a self-adjoint operator. If $f$ is an  arbitrary random variable we can write $f=f^+-f^-$ where $f^+(x)=\max\sqbrac{f(x),0}$ and
$f^-(x)=-\min\sqbrac{f(x),0}$. Then we have that $f^+,f^-\ge 0$ and we define the \textit{quantization}
$\fhat =f^{+\wedge}-f^{-\wedge}$. Again, $\fhat$ is a bounded self-adjoint operator on $H$. According to the usual formalism, we can interpret $\fhat$ as an observable for a quantum system.

\begin{thm}       
\label{thm31}
{\rm (a)}\enspace For any $A\in\ascript$, $\chihat _A=\ket{\chi _A}\bra{\chi _A}=\mu (A)$.
{\rm (b)}\enspace For any $\alpha\in\real$, $(\alpha f)^\wedge =\alpha\fhat$.
{\rm (c)}\enspace If $0\le f_1\le f_2\le\cdots$ is an increasing sequence of random variables converging pointwise to a random variable $f$, then $\fhat _i$ converges to $\fhat$ in the operator norm topology.
{\rm (d)}\enspace If $f,g,h$ are random variables with mutually disjoint support, then
\begin{equation}         
\label{eq32}
(f+g+h)^\wedge =(f+g)^\wedge+(f+h)^\wedge +(g+h)^\wedge -\fhat -\ghat -\hhat
\end{equation}
\end{thm}
\begin{proof}
(a)\enspace For $A\in\ascript$ we have that $\min\sqbrac{\chi _A(x),\chi _A(y)}=\chi _A(x)\chi _A(y)$. Hence, for
$g\in H$ we obtain
\begin{equation*}
(\chihat _Ag)(y)=\int\chi _A(x)\chi _A(y)g(x)d\nu (x)=\int _Agd\nu\chi _A(y)=\sqbrac{\mu (A)g}(y)
\end{equation*}
(b) If $\alpha\ge 0$ and $f\ge 0$, then clearly $(\alpha f)^\wedge =\alpha\fhat$. If $\alpha\ge 0$ and $f$ is a random variable, then
\begin{equation*}
\alpha f=(\alpha f)^+-(\alpha f)^-=\alpha f^+-\alpha f^-
\end{equation*}
Hence,
\begin{equation*}
(\alpha f)^\wedge =(\alpha f^+)^\wedge -(\alpha f^-)^\wedge
  =\alpha f^{+\wedge}-\alpha f^{-\wedge}=\alpha f^\wedge
\end{equation*}
If $\alpha <0$ and $f$ is a random variable, then
\begin{equation*}
\alpha f=(\alpha f)^+-(\alpha f)^-=\ab{\alpha}f^--\ab{\alpha}f^+
\end{equation*}
Hence,
\begin{equation*}
(\alpha f)^\wedge =\paren{\ab{\alpha}f^-}^\wedge -\paren{\ab{\alpha}f^+}^\wedge
  =\ab{\alpha}f^{-\wedge}-\ab{\alpha}f^{+\wedge}=-\ab{\alpha}\fhat =\alpha\fhat
\end{equation*}
(c) For any $g\in H$ we have
\begin{align*}
\ab{(\fhat -\fhat _i)g(y)}&=\ab{\int\brac{\min\sqbrac{f(x),f(y)}-\min\sqbrac{f_i(x),f_i(y)}}g(x)d\nu (x)}\\
  &\le\int\ab{\min\sqbrac{f(x),f(y)}-\min\sqbrac{f_i(x),f_i(y)}}\ab{g(x)}d\nu (x)\\
  &\le\int\sqbrac{\ab{f(y)-f_i(y)}+\ab{f(x)-f_i(x)}}\ab{g(x)}d\nu (x)\\
  &\le\sqbrac{f(y)-f_i(y)}\doubleab{g}+\doubleab{f-f_i}\,\doubleab{g}
\end{align*}
Squaring, we obtain
\begin{align*}
\ab{(\fhat -\fhat _i)g(y)}^2&\le\paren{\sqbrac{f(y)-f_i(y)}+\doubleab{f-f_i}}^2\doubleab{g}^2\\
  &=\paren{\sqbrac{f(y)\!-\!f_i(y)}^2+\doubleab{f\!-\!f_i}^2+2\sqbrac{f(y)\!-\!f_i(y)}\doubleab{f\!-\!f_i}}\doubleab{g}^2
\end{align*}
Hence,
\begin{equation*}
\doubleab{\fhat -\fhat _i)g}^2\le 4\doubleab{f-f_i}^2\doubleab{g}^2
\end{equation*}
We conclude that
\begin{equation*}
\doubleab{\fhat -\fhat _i}\le 2\doubleab{f-f_i}
\end{equation*}
By the monotone convergence theorem, $\lim\doubleab{f-f_i}=0$ and the result follows.\newline
(d)\enspace If $f\ge 0$ is a simple function $f=\sum _{i=1}^n\alpha _i\chi _{A_i}$, $\alpha _i>0$, then
\begin{equation*}
\min\sqbrac{f(x),f(y)}=\sum _{i,j=1}^n\min (\alpha _i,\alpha _j)\chi _{A_i}(x)\chi _{A_j}(y)
\end{equation*}
Hence, for any $u\in H$ we have
\begin{align}         
\label{eq33}
(\fhat u)(y)&=\int\sum _{i,j=1}^n\min (\alpha _i,\alpha _j)\chi _{A_i}(x)\chi _{A_j}(y)u(x)d\nu (x)\notag\\
  &=\sum _{i,j=1}^n\min (\alpha _i,\alpha _j)\int _{A_i}ud\nu\chi _{A_j}(y)
\end{align}
Suppose $g\ge 0$ and $h\ge 0$ are simple functions with $g=\sum\beta _i\chi _{B_i}$ and
$h=\sum\gamma _i\chi _{C_i}$, $\alpha _i,\beta _i,\gamma _i>0$. Assuming that $f,g,h$ have disjoint support, it follows that $\cup A_i$, $\cup B_i$, $\cup C_i$ are mutually disjoint. Employing the notation
\begin{equation*}
I(a,b)=\sum _{i,j}\min (a_i,b_j)\int _{A_i}ud\nu\chi _{A_j}
\end{equation*}
As in \eqref{eq33} we obtain
\begin{align*}
(f+g)^\wedge&u+(f+h)^\wedge u+(g+h)^\wedge u-\fhat u-\ghat u-\hhat u\\
  &=I(\alpha ,\alpha )+I(\beta ,\beta )+I(\gamma ,\gamma )+I(\alpha ,\beta )
  +I(\beta ,\alpha )+I(\alpha ,\gamma )\\
  &\quad +I(\gamma ,\alpha )+I(\beta ,\gamma )+I(\gamma ,\beta)\\
  &=(f+g+h)^\wedge u
\end{align*}
We conclude that \eqref{eq32} holds for simple nonnegative random variables with disjoint support. Now suppose $f,g,h$ are arbitrary nonnegative random variables with disjoint support. Then there exist increasing sequences $f_i,g_i,h_i$ of nonnegative simple random variables converging pointwise to $f$, $g$ and $h$, respectively. Since \eqref{eq32} holds for $f_i,g_i,h_i$, applying (c) shows that \eqref{eq32} holds for $f,g,h$. Finally, let $f,g,h$ be arbitrary random variables with disjoint support. It is easy to check that $(f+g)^+=f^++g^+$,
$(f+g)^-=f^-+g^-$, $(f+g+h)^+=f^++g^++h^+$, etc. Then \eqref{eq32} becomes
\begin{align*}
(f^+&+g^++h^+)^\wedge -(f^-+g^-+h^-)^\wedge\\
  &=(f^++g^+)^\wedge -(f^-+g^-)^\wedge +(f^++h^+)^\wedge -(f^-+h^-)^\wedge +(g^++h^+)^\wedge\\
  &\quad -(g^-+h^-)^\wedge -f^{+\wedge}+f^{-\wedge}-g^{+\wedge}+g^{-\wedge}-h^{+\wedge}+h^{-\wedge}
\end{align*}
But this follows from our previous work because $f^+,g^+,h^+$ and $f^-,g^-,h^-$ are nonnegative and have disjoint support so \eqref{eq32} holds for $f^+,g^+,h^+$ and also for $f^-,g^-,h^-$.
\end{proof}

The next result shows that $f\mapsto\fhat$ preserves positivity.

\begin{thm}       
\label{thm32}
If $f\ge 0$ is a random variable, then $\fhat$ is a positive operator.
\end{thm}
\begin{proof} Suppose $f\ge 0$ is a simple function with $f=\sum _{i=1}^n\alpha _i,\chi _{A_i}$ where
$0\le\alpha _1<\alpha _2<\cdots <\alpha _n$, $A_i\cap A_j=\emptyset$, $i\ne j$ and $\cup A_i=\Omega$. If $u\in H$ and $B_j=\bigcup\limits _{i=j}^nA_i$, then by \eqref{eq33} we have
\begin{align}         
\label{eq34}
\elbows{\fhat u,u}&=\sum _{i,j=1}^n\min (\alpha _i,\alpha _j)\int _{A_i}ud\nu\int _{A_j}\ubar d\nu\notag\\
  &=\alpha _1\sqbrac{\int _{A_1}ud\nu\int _{B_1}\ubar d\nu +\int _{B_2}ud\nu\int _{A_1}\ubar}\notag\\
  &\quad +2\alpha _2\rmre\int _{A_2}ud\nu\int _{B_2}\ubar d\nu
  +2\alpha _3\rmre\int _{A_3}ud\nu\int _{B_3}\ubar d\nu\notag\\
  &\quad +\cdots +2\alpha _n\int _{A_n}ud\nu\int _{B_n}\ubar d\nu\notag\\
  &=\alpha _1\sqbrac{\ab{\int _{B_1}ud\nu}^2-\ab{\int _{B_2}ud\nu}^2}
  +\alpha _2\sqbrac{\ab{\int _{B_2}ud\nu}^2-\ab{\int _{B_3}ud\nu}^2}\notag\\
  &\quad +\cdots +\alpha _n\ab{\int _{B_n}ud\nu}^2\notag\\
  &=\alpha _1\ab{\int _{B_1}ud\nu}^2+(\alpha _2-\alpha _1)\ab{\int _{B_2}ud\nu}^2+(\alpha _3-\alpha _2)
  \ab{\int _{B_3}ud\nu}^2\notag\\
  &\quad +\cdots +(\alpha _n-\alpha _{n-1})\ab{\int _{B_n}ud\nu}^2\ge 0
\end{align}
Hence, $\fhat$ is a positive operator. If $f\ge 0$ is an arbitrary random variable, then there exists an increasing sequence of simple functions $f_i\ge 0$ converging pointwise to $f$. Since $\fhat _i$ are positive, it follows from Theorem~\ref{thm31}(c) that $\fhat$ is positive.
\end{proof}

A random variable $f$ satisfying $0\le f\le 1$ is called a \textit{fuzzy} (or \textit{unsharp}) \textit{event} while functions
$\chi _A$, $A\in\ascript$, are called \textit{sharp events} \cite{dp00}. If $\nu (A)\ne 0$, denote the projection onto
$\rmspan (\chi _A)$ by $P(A)$. Theorem~\ref{thm31}(a) shows that $\chihat _A=\nu (A)P(A)$. Thus, quantization takes sharp events to constants times projections. An operator $T$ on $H$ satisfying $0\le T\le I$ is called an
\textit{effect} \cite{dp00}. Since $\doubleab{\fhat}\le\doubleab{f}$, it follows from Theorem~\ref{thm32} that quantization takes fuzzy events into effects.

Let $\rho$ be a density operator on $H$ and let $\mu _\rho (A)=\rmtr\sqbrac{\rho\mu (A)}$ be the corresponding
$q$-measure. If $f$ is a random variable, we define the $q$-\textit{integral} (or $q$-\textit{expectation}) of $f$ with respect to $\mu _\rho$ as
\begin{equation*}
\int fd\mu _\rho =\rmtr (\rho\fhat )
\end{equation*}
As usual, for $A\in\ascript$ we define
\begin{equation*}
\int _Afd\mu _\rho =\int\chi _Afd\mu _\rho
\end{equation*}
The next result follows from Theorems~\ref{thm31} and \ref{thm32}.

\begin{cor}       
\label{cor33}
{\rm (a)}\enspace For all $A\in\ascript$, we have $\int\chi _Ad\mu _\rho =\mu _\rho (A)$.
{\rm (b)}\enspace For all $\alpha\in\real$, we have $\int\alpha fd\mu _\rho=\alpha\int fd\mu _\rho$.
{\rm (c)}\enspace If $f_i\ge 0$, is an increasing sequence of random variables converging to a random variable $f$, then $\lim\int f_id\mu _\rho =\int fd\mu _\rho$.
{\rm (d)}\enspace If $f\ge 0$, then $\int fd\mu _\rho\ge 0$.
{\rm (e)}\enspace If $f,g,h$ are random variables with disjoint support, then
\begin{align*}
\int (f+g+h)d\mu _\rho&=\int (f+g)d\mu _\rho +\int (f+h)d\mu _\rho +\int (g+h)d\mu _\rho\\
  &\quad -\int fd\mu _\rho -\int gd\mu _\rho -\int hd\mu _\rho
\end{align*}
{\rm (f)}\enspace If $A,B,C\in\ascript$ are mutually disjoint, then 
\begin{align*}
\int _{A\cup B\cup C}fd\nu&=\int _{A\cup B}fd\mu _\rho +\int _{A\cup C}fd\mu _\rho +\int _{B\cup C}fd\mu _\rho\\
  &\quad -\int _Afd\mu _\rho -\int _Bfd\mu _\rho -\int _Cfd\mu _\rho
\end{align*}
\end{cor}

The following result is called the \textit{tail-sum formula}

\begin{thm}       
\label{thm34}
If $f\ge 0$ is a random variable, then
\begin{equation*}
\int fd\mu _\rho =\int _0^\infty\mu _\rho\brac{x\colon f(x)>\lambda}d\lambda
\end{equation*}
where $d\lambda$ denotes Lebesgue measure on $\real$.
\end{thm}
\begin{proof}
Let $f\ge 0$ be a simple function with $f=\sum\alpha _i\chi _{A_i}$, $0\le\alpha _1<\alpha _2<\cdots <\alpha _n$. Let
$u\in H$ with $\doubleab{u}=1$ and let $\rho$ be the density operator given by $\rho =\ket{u}\bra{u}$. Of course,
$\rho$ is a pure state. Then
\begin{equation*}
\mu _\rho (A)=\ab{\elbows{u,\chi _A}}^2=\ab{\int _Aud\nu}^2
\end{equation*}
Using the notation of Theorem~\ref{thm32}, we have
\begin{align*}
\mu _\rho\brac{x\colon f(x)>\lambda}&=\mu _\rho (B_{i+1})\hspace{2pc}\hbox{for }\alpha _i<\lambda\le\alpha _{i+1}\\
\mu _\rho\brac{x\colon f(x)>\lambda}&=\mu _\rho (B_1)\hspace{2.5pc}\hbox{for }0\le\lambda\le\alpha _1\\
\mu _\rho\brac{x\colon f(x)>\lambda}&=0\hspace{5pc}\hbox{for }\alpha _n<\lambda
\end{align*}
Moreover, we have
\begin{equation*}
\mu _\rho (B_{i+1})=\ab{\elbows{u,\sum _{j=i+1}^n\chi _{A_j}}}^2=\ab{\sum _{j=i+1}^n\int _{A_j}ud\nu}^2
\end{equation*}
We conclude that
\begin{align}         
\label{eq35}
\int _0^\infty\mu _\rho&\brac{x\colon f(x)>\lambda}d\lambda\notag\\
  &=\int _0^{\alpha _1}\mu _\rho\brac{x\colon f(x)>\lambda}d\lambda
  +\int _{\alpha _1}^{\alpha _2}\mu _\rho\brac{x\colon f(x)>\lambda}d\lambda\notag\\
  &\quad +\cdots +\int _{\alpha _{n-1}}^{\alpha _n}\mu _\rho\brac{x\colon f(x)>\lambda}d\lambda\notag\\
  &=\alpha _1\mu _\rho (B_1)+(\alpha _2-\alpha _1)\mu _\rho (B_2)+(\alpha _3-\alpha _2)\mu _\rho (B_3)\notag\\
  &\quad +\cdots +(\alpha _n-\alpha _{n-1})\mu _\rho (B_n)\notag\\
  &=\alpha\ab{\int _{B_1}ud\nu}^2+(\alpha _2-\alpha _1)\ab{\int _{B_2}ud\nu}^2\notag\\
  &\quad +\cdots +(\alpha _n-\alpha _{n-1})\ab{\int _{B_n}ud\nu}^2
\end{align}

Comparing \eqref{eq34} and \eqref{eq35} gives
\begin{equation*}
\int fd\mu _\rho =\elbows{\fhat u,u}=\int _0^\infty\mu _\rho\brac{x\colon f(x)>\lambda}d\lambda
\end{equation*}
Applying Theorem~\ref{thm23}(b) and Corollary~\ref{cor33}(c) we conclude that the result holds because a mixed state is a convex combination of pure states.
\end{proof}

Applying Theorem~\ref{thm34} we have for an arbitrary random variable $f$ that
\begin{equation*}
\int fd\mu _\rho =\int _0^\infty\mu _\rho\brac{x\colon f(x)>\lambda}d\lambda
  -\int _0^\infty\mu _\rho\brac{x\colon f(x)<-\lambda}d\lambda
\end{equation*}
This result shows that the present definition of a $q$-integral reduces to the definition studied previously
\cite{gud091, gud092}.

\section{Finite Unitary Systems} 
This section discusses a physical example that illustrates the theory of Section~2. A \textit{finite unitary system} is a collection of unitary operators $U(s,r)$, $r\le s\in\positive$, on $\complex ^m$ such that $U(r,r)=I$ and
$U(t,r)=U(t,s)U(s,r)$ for all $r\le s\le t\in\positive$. These operators describe the evolution of a finite-dimensional quantum system in discrete steps from time $r$ to time $s$. If $\brac{U(s,r)\colon r\le s}$ is a finite unitary system, then we have the unitary operators $U(n+1,n)$, $n\in\positive$, such that
\begin{equation}         
\label{eq41}
U(s,r)=U(s,s-1)U(s-1,s-2)\cdots U(r+1,r)
\end{equation}
Conversely, if $U(n+1,n)$, $n\in\positive$, are unitary operators on $\complex ^m$, then defining $U(r,r)=I$ and for
$r<s$ defining $U(s,r)$ by \eqref{eq41} we have the finite unitary system $\brac{U(s,r)\colon r\le s}$.

In the sequel, $\brac{U(s,r)\colon r\le s}$ will be a fixed finite unitary system on $\complex ^m$. Suppose the evolution of a particle is governed by $U(s,r)$ and the particle's position is at one of the points $0,1,\ldots ,m-1$. We call the elements of $S=\brac{0,1,\ldots ,m-1}$ \textit{sites} and the infinite strings
$\gamma =\gamma _1\gamma _1\gamma _2\cdots$, $\gamma _i\in S$ are called \textit{paths}. The paths represent particle trajectories and the \textit{path} or \textit{sample space} is
\begin{equation*}
\Omega =\brac{\gamma\colon\gamma\hbox{ a path}}
\end{equation*}
The finite strings $\gamma =\gamma _0\gamma _1\cdots\gamma _n$ are $n$-\textit{paths} and
\begin{equation*}
\Omega _n=\brac{\gamma\colon\gamma\hbox{ an $n$-path}}
\end{equation*}
is the $n$-\textit{path} or $n$-\textit{sample space}. The $n$-paths represent time-$n$ truncated particle trajectories. Notice that the cardinality $\ab{\Omega _n}=m^{n+1}$. The power set $\ascript _n=2^{\Omega _n}$ is the set of
$n$-\textit{events}. Letting $\nu _n$ be the uniform distribution $\nu _n(\gamma )=1/m^{n+1}$,
$\gamma\in\Omega _n$, $(\Omega _n,\ascript _n,\nu _n)$ becomes a probability space.

We call $\complex ^m$ the \textit{position Hilbert space} Let $e_0,\ldots ,e_{m-1}$ be the standard basis for
$\complex ^m$ and let $P(i)=\ket{ e_i}\bra{e_i}$, $i=0,1,\ldots ,m-1$ be the corresponding projection operators. For
$\gamma\in\Omega _n$ the operator $C_n(\gamma )$ on $\complex ^m$ that describes this trajectory is
\begin{equation*}
C_n(\gamma )=P(\gamma _n)U(n,n-1)P(\gamma _{n-1})U(n-1,n-2)\cdots P(\gamma _1)U(1,0)P(\gamma _0)
\end{equation*}
Letting
\begin{equation}         
\label{eq42}
b(\gamma )=\elbows{e_{\gamma _n},U(n,n-1)e_{\gamma _{n-1}}}\cdots
  \elbows{e_{\gamma _2},U(2,1)e_{\gamma _1}}\elbows{e_{\gamma _1},U(1,0)e_{\gamma _0}}
\end{equation}
we have that
\begin{equation}         
\label{eq43}
C_n(\gamma )=b(\gamma )\ket{e_{\gamma _n}}\bra{e_{\gamma _0}}
\end{equation}
If $\psi\in\complex ^m$ is a unit vector, the complex number $a_\psi (\gamma )=b(\gamma )\psi (\gamma _0)$ is the \textit{amplitude} of $\gamma$ with initial distribution (state) $\psi$. It is easy to show that
\begin{equation}         
\label{eq44}
\sum _{\gamma\in\Omega _n}\ab{a_\psi (\gamma )}^2=1
\end{equation}
Moreover, for all $\gamma ,\gamma '\in\Omega _n$ we have
\begin{equation}         
\label{eq45}
C_n(\gamma ')^*C_n(\gamma )
  =\overline{b(\gamma ')}b(\gamma )\ket{e_{\gamma '_0}}\bra{e_{\gamma _0}}\delta _{\gamma _n,\gamma '_n}
\end{equation}
The operator $C_n(\gamma ')^*C_n(\gamma )$ describes the interference between the paths $\gamma$ and
$\gamma '$.

For $A\in\ascript _n$, the \textit{class operator} $C_n(A)$ is defined by
\begin{equation*}
C_n(A)=\sum _{\gamma\in A}C_n(\gamma )
\end{equation*}
It is clear that $A\mapsto C_n(A)$ is an operator-valued measure on the algebra $\ascript _n$ satisfying
$C_n(\Omega _n)=U(n,0)$. The \textit{decoherence functional}
$\Delta _n\colon\ascript _n\times\ascript _n\to\complex$ is given by
$\Delta _n(A,B)=\elbows{C_n^*(A)C_n(B)\psi ,\psi}$ where $\psi\in\complex ^m$ is the initial state. It is clear that
$A\mapsto\Delta _n(A,B)$ is a complex measure on $\ascript _n$ for any $B\in\ascript _n$,
$\Delta _n(\Omega _n,\Omega _n)=1$ and it is well-known that if $A_1,\ldots ,A_r\in\ascript$, then
$\Delta _n(A_i,A_j)$ is an $r\times r$ positive semidefinite matrix \cite{gt09, sor941, sor942, sor071}. Corresponding to an initial state
$\psi\in\complex ^m$, $\doubleab{\psi}=1$, the $n$-decoherence matrix is
\begin{equation*}
\Delta _n(\gamma ,\gamma ')=\Delta _n\paren{\brac{\gamma},\brac{\gamma '}}
\end{equation*}
Applying \eqref{eq45} we have
\begin{align}         
\label{eq46}
\Delta _n(\gamma ,\gamma ')&=\elbows{C_n(\gamma ')^*C_n(\gamma )\psi ,\psi}
  =\overline{b(\gamma ')}b(\gamma )\overline{\psi (\gamma '_0)}\psi (\gamma _0)\delta _{\gamma _n,\gamma '_n}
  \notag\\
  &=\overline{a_\psi (\gamma')}a_\psi (\gamma )\delta _{\gamma _n,\gamma '_n}
\end{align}

Define the $n$-\textit{path Hilbert space} $H_n=(\complex ^m)^{\otimes (n+1)}$. For $\gamma\in\Omega _n$ we associate the unit vector in $H_n$ given by
\begin{equation*}
e_{\gamma _n}\otimes e_{\gamma _{n-1}}\otimes\cdots\otimes e_{\gamma _0}
\end{equation*}
We can think of $H_n$ as $\brac{\phi\colon\Omega _n\to\complex}$ with the usual inner product and
$\Delta _n(\gamma ,\gamma ')$ corresponds to the operator
\begin{equation*}
(\Delta _n\phi )(\gamma )=\sum _{\gamma '}\Delta _n(\gamma ,\gamma ')\psi (\gamma ')
\end{equation*}
Since $\Delta _n(\gamma ,\gamma ')$ is a positive semidefinite matrix, $\Delta _n$ is a positive operator on $H_n$ and by \eqref{eq44} and \eqref{eq46} we have
\begin{equation*}
\rmtr (\Delta _n)=\sum _\gamma\Delta _n(\gamma ,\gamma )=\sum _\gamma\ab{a_\psi (\gamma )}^2=1
\end{equation*}
We conclude that $\Delta _n$ is a state on $H_n$.

\begin{lem}       
\label{lem41}
The decoherence functional satisfies
\begin{equation*}
\Delta _n(A,B)=\rmtr\paren{\ket{\chi _A}\bra{\chi _B}\Delta _n}
\end{equation*}
for all $A,B\in\ascript$.
\end{lem}
\begin{proof}
By the definition of the trace we have
\begin{align*}
\rmtr\paren{\ket{\chi _A}\bra{\chi _B}\Delta _n}
  &=\sum _\gamma\elbows{\ket{\chi _A}\bra{\chi _B}\Delta _n\gamma ,\gamma}\\
  &=\sum _\gamma\elbows{\Delta _n\gamma ,\chi _B}\elbows{\chi _A,\gamma}
  =\sum _{\gamma\in A}\elbows{\Delta _n\gamma ,\chi _B}\\
  &=\sum\brac{\elbows{\Delta _n\gamma ,\gamma '}\colon\gamma\in A,\gamma '\in B}\\
  &=\sum\brac{\Delta _n(\gamma ,\gamma ')\colon\gamma\in A,\gamma '\in B}\\
  &=\Delta _n(A,B)\qedhere
\end{align*}
\end{proof}

Lemma~\ref{lem41} shows that the decoherence functional as it is usually defined coincides with the decoherence functional of Section~2. Moreover, the usual $q$-measure $\mu _{\Delta _n}(A)=\Delta _n(A,A)$ coincides with the $q$-measure of Section~2. We have only discussed the time-$n$ truncated path space $\Omega _n$. The infinite time path space $\Omega$ is of primary interest, but its study is blocked by mathematical difficulties. It is hoped that the present structure will help to make progress in overcoming these difficulties.


\begin{thebibliography}{99}
\bibitem{dgt08}F~Dowker and Y.~Ghazi-Tabatabai, 
Dynamical wave function collapse models in quantum measure theory, \textit{J. Phys. A} \textbf{41} (2008), 105301.
\bibitem{dp00}A.~Dvure\v censkij and S.~Pulmannov\'a, 
\textit{New Trends in Quantum Structures}, Kluwer, Dordrecht, The Netherlands, 2000.
\bibitem{gt09}Y.~Ghazi-Tabatabai, 
Quantum measure theory: a new interpretation, arXiv: quant-ph (0906:0294) 2009.
\bibitem{gtw09}Y.~Ghazi-Tabatabai and P.~Wallden, Dynamics and predictions in the co-event interpretation, 
\textit{J. Phys. A} \textbf{42} (2009), 235303.
\bibitem{gud101}S.~Gudder, Quantum measure theory, \textit{Math. Slovaca} \textbf{60}, (2010), 681--700.
\bibitem{gud091}S.~Gudder, Quantum measure and integration theory, \textit{J. Math. Phys.} \textbf{50},
(2009), 123509.
\bibitem{gud092}S.~Gudder, Quantum integrals and anhomomorphic logics, arXiv: quant-ph (0911.1572), 2009.
\bibitem{gud102}S.~Gudder, Quantum measures and the coevent interpretation, \textit{Rep. Math. Phys.} 
\textbf{67} (2011), 137--156.
\bibitem{mocs05}X.~Martin, D.~O'Connor and R.~Sorkin, Random walk in generalized quantum theory,
\textit{Physic Rev D} \textbf{71} (2005), 024029.
\bibitem{sal02}R.~Salgado, Some identities for quantum measures and its generalizations,
\textit{Mod. Phys. Letts. A} \textbf{17} (2002), 711-728.
\bibitem{sor941}R.~Sorkin, 
Quantum mechanics as quantum measure theory, \textit{Mod. Phys. Letts.~A} \textbf{9} (1994), 3119--3127.
\bibitem{sor942}R.~Sorkin, 
Quantum measure theory and its interpretation, in 
\textit{Proceedings of the 4th Drexel Symposium on quantum nonintegrability}, eds. D.~Feng and B.-L.~Hu,
1994, 229--251.
\bibitem{sor071}R.~Sorkin, 
Quantum mechanics without the wave function, \textit{Mod. Phys. Letts.~A} \textbf{40} (2007), 3207--3231.
\bibitem{sor072}R.~Sorkin, 
An exercise in ``anhomomorphic logic'', \textit{J.~Phys.: Conference Series} (JPCS) \textbf{67}, 012018 (2007).
\bibitem{sw10}S.~Surya and P.~Wallden, 
Quantum covers in quantum measure theory, \textit{Found.} \textbf{40} (2010) 585--606.

\end{thebibliography}
\end{document}